\documentclass[a4paper,twoside]{article}

\usepackage{epsfig}
\usepackage{subcaption}
\usepackage{calc}
\usepackage{amssymb}
\usepackage{amstext}
\usepackage{amsmath}
\usepackage{amsthm}
\usepackage{multicol}
\usepackage{pslatex}
\usepackage{apalike}
\usepackage{algorithm2e}
\usepackage[bottom]{footmisc}
\usepackage{SCITEPRESS}     % Please add other packages that you may need BEFORE the SCITEPRESS.sty package.

\usepackage{balance}

\usepackage{mathtools}
\usepackage[acronym]{glossaries}
\newacronym{MPC}{MPC}{Model Predictive Control}
\newacronym{RMPC}{RMPC}{Robust Model Predictive Control}
\newacronym{DoS}{DoS}{Denial of Service}
\newacronym{ACC}{ACC}{Adaptive Cruise Control}
\newacronym{ROA}{ROA}{Region of Attraction}
\newacronym{IPOPT}{IPOPT}{interior point method}

\newacronym{SOL}{SOL}{Region of Attraction}

\newtheorem{lemma}{Lemma}
\newtheorem{remark}{Remark}

\newtheorem{assumption}{Assumption}

\begin{document}

\title{An Efficient Resilient MPC Scheme via Constraint Tightening against Cyberattacks: Application to Vehicle Cruise Control}

\author{\authorname{Milad Farsi\sup{1}, Shuhao Bian\sup{1}, Nasser L. Azad \sup{1}, Xiaobing Shi \sup{2} and Andrew Walenstein \sup{2}}
\affiliation{\sup{1}Department of Systems Design Engineering, University of Waterloo, Waterloo, Canada}
\affiliation{\sup{2}BlackBerry Limited, Waterloo, Canada}
\email{\{mfarsi, s6bian, nlashgarianazad\}@uwaterloo.ca, \{xiashi, awalenstein\}@blackberry.com}
}

% Milad Farsi (mfarsi@uwaterloo.ca)
% Co-Authors:	Shuhao Bian (s6bian@uwaterloo.ca)
% Nasser L. Azad (nlashgarianazad@uwaterloo.ca)
% Xiaobing Shi (xiashi@blackberry.com)
% Andrew Walenstein (awalenstein@blackberry.com)


% \abstract{The abstract should summarize the contents of the paper and should contain at least 70 and at most 200 words. The text must be set to 9-point font size.}
\abstract{We propose a novel framework for designing a resilient \gls{MPC} targeting uncertain linear systems under cyber attack. Assuming a periodic attack scenario, we model the system under \gls{DoS} attack, also with measurement noise, as an uncertain linear system with parametric and additive uncertainty. To detect anomalies, we employ a Kalman filter-based approach. Then, through our observations of the intensity of the launched attack, we determine a range of possible values for the system matrices, as well as establish bounds of the additive uncertainty for the equivalent uncertain system. Leveraging a recent constraint tightening robust \gls{MPC} method, we present an optimization-based resilient algorithm. Accordingly, we compute the uncertainty bounds and corresponding constraints offline for various attack magnitudes. Then, this data can be used efficiently in the \gls{MPC} computations online. We demonstrate the effectiveness of the developed framework on the \gls{ACC} problem.}

\onecolumn 
\maketitle
\normalsize 

\setcounter{footnote}{0} \vfill

\glsresetall
\allowbreak
\allowdisplaybreaks
\section{\uppercase{Introduction}}
\label{sec:introduction}

Resilient control refers to the capability of a control system to maintain stable and optimal performance despite cyber-attacks \cite{sandberg2022secure}, disturbances, uncertainties, and faults. Traditional control systems assume ideal conditions, leading to performance degradation or failure during unexpected events. Resilient controllers enhance robustness and adaptability, especially against cyber-attacks in critical systems. In the context of modern vehicles vulnerable to cyber threats, successful attacks can cause loss of control, safety compromises, and harm to passengers \cite{ju2022survey}. Resilient control techniques detect, mitigate, and recover from cyber-attacks, preserving vehicle functionality and safety in adverse conditions.

\gls{DoS} as one of the well-known cyber attacks have become increasingly prevalent in today's digital landscape that can gravely affect modern vehicle systems \cite{biron2018real}. These attacks aim to disrupt or disable the targeted system's services or resources, making them unavailable to legitimate users. Therefore, different techniques are employed in the literature to mitigate potential damages caused by such attacks. Game theory provides a framework for modeling strategic interactions and decision-making processes during cyber attacks \cite{gupta2016dynamic,huang2020dynamic}. Moreover, event-triggered control methods have been popular considering their advantages in cyber-physical systems, including vehicle control \cite{xiao2020resilient,wu2022recent}.    

\gls{RMPC} as a subcategory of \gls{MPC} techniques is a powerful control framework that excels at handling uncertainty and disturbance in real-world applications \cite{bemporad2007robust}. The inherent ability of \gls{RMPC} to explicitly account for uncertainties makes it particularly well-suited for complex systems operating in dynamic environments. \gls{RMPC} encompasses various approaches to handle uncertainties and disturbances in control systems. Min-Max \gls{RMPC} approaches formulate the control problem as a min-max optimization, where the objective is to minimize online the worst-case performance subject to constraints \cite{raimondo2009min}. However, these techniques can involve overly expensive computations. 
\begin{figure*}[!ht]
  \centering
   {\epsfig{file = fig_acc_example.pdf, width = 11cm}}
  \caption{A view of the \gls{ACC} problem}
  \label{fig:acc_example}
 \end{figure*}
Tube-based \gls{RMPC} constructs an invariant set, known as the robust tube, that captures the possible system trajectories considering the uncertainties \cite{langson2004robust,sakhdari2018adaptive}. By formulating the optimization problem within this tube, the system stability and constraint satisfaction are obtained. In \cite{mayne2005robust}, the authors address \gls{RMPC} problem in the presence of bounded disturbances for constrained linear discrete-time systems. 

In \cite{aubouin2022resilient}, to address the resilience issue in maintaining system operation under repeated \gls{DoS} attacks, the concept of $\mu$-step Robust Positively Invariant ($\mu$-RPI) sets is introduced. These sets aim to restrict the impact of attacks, ensuring that any deviation from nominal operation remains limited in time and/or magnitude. Although such approaches offer different perspectives on robust control design and enable the handling of uncertainties and disturbances efficiently, they may result in rather conservative and computationally expensive solutions. 

Constraint-tightening techniques involve iteratively refining the constraints in an optimization problem to enforce robustness \cite{kohler2018novel}. In \cite{bujarbaruah2021simple}, the authors demonstrated that by selecting appropriate terminal constraints and employing an adaptive horizon strategy, constraint tightening may not necessarily result in excessively conservative behavior where its \gls{ROA} can be as large as $98\%$ of the tube-based techniques, such as \cite{langson2004robust}. More importantly, it can run 15x faster. 

Building upon this \gls{RMPC} approach, in this paper, we develop an efficient resilient control framework against a class of cyber-attacks that can be potentially employed in real-time as an alternative to the current \gls{MPC} implementations. Our approach involves the computation of a set of uncertain models that encompasses different levels of the strength of \gls{DoS} attacks, as well as accounting for potential noise and unmodeled dynamics. Through an iterative scheme, we employ the Kalman filter to detect the occurrence of an attack. Subsequently, in the control loop, we estimate the intensity of the attack and adaptively evaluate the control based on the models pre-computed for different circumstances specifically.
Hence, our contributions include: obtaining an overapproximate model with additive and parametric uncertainties based on a practical problem formulation, developing a resilient control framework that can be employed as an extension of \cite{bujarbaruah2021simple} against \gls{DoS}, and validating it on the Adaptive Cruise Control (ACC) problem, illustrated in Fig.\ref{fig:acc_example}.

The rest of the paper is organized as follows. In Section \ref{sec:Problem Formulation}, we formulate the problem. Section \ref{sec:Resilient Framework} presents the resilient control framework and outlines the algorithms designed based on the obtained results. In Section \ref{sec:Anomaly Detection}, we summarize a commonly used anomaly detection method. In Section \ref{sec:Simulation Results}, we present a case study to validate and compare the proposed approach in the simulation environment.  

\paragraph{Notations.}
We denote $n$-dimensional Euclidean space by $\mathbb{R}^n$, and the space of positive reals with the subscript as $\mathbb{R}^n_+$. We further denote by $|X|$ the absolute value of a variable $X$, where for non-scalars, it represents the component-wise absolute value. $||x||$ denotes the 2-norm of a vector $x \in \mathbb{R}^n$. A diagonal square matrix $A$ with elements $a_1,\dots,a_n$ on the diagonal is shortened as $A=\text{diag}([a_1,\dots,a_n])$. We use the upper script as $x^k$ for discrete-time signals, and $x^{k|j}$ represents the estimation of $x^k$ at the time j. $diag(x_1,x_2,...,x_n)$ denotes a diagonal matrix of the elements $x_1,x_2,...,x_n$. The set $G[a,b]$ represents a grid with the bounds $a$ and $b$. 

\section{\uppercase{Problem Formulation}}
\label{sec:Problem Formulation}

In this section, considering a class of systems, we formulate the effect of the \gls{DoS} attack. Accordingly, we define the problem of interest.

Consider the following system
\begin{align} \label{system}
    &\dot x =Ax+\Delta(x)+Bu
\end{align}
where $x\in D\subset {\rm I\!R}^n$ and $u\in U \subset{\rm I\!R}^m$ are respectively the state and control input, and take values on the compact convex sets $D$ and $U$. System matrices are given as $A \in {\rm I\!R}^{n\times n}$, and $B\in {\rm I\!R}^{n\times m}$.  Furthermore, the unmodeled dynamics are given by $\Delta:D\rightarrow {\rm I\!R}^n$.

\begin{assumption} \label{assu_lip}
\leavevmode
    Elements of $\Delta(x)$ are assumed to be a Lipschitz continuous function of the state, i.e. $\exists\eta\in{\rm I\!R}^{n}_+$ such that we have
	\begin{align*}
	{| \Delta_i(x_0)-\Delta_i(y_0)|} \leq \eta_{i} {\lVert x_0-y_0\rVert},
	\end{align*} 
	for any $x_0,y_0 \in D$, where $i \in \{1,\dots,n\}$.
\end{assumption}

While more general classes of dynamical systems do exist, the specific formulation we adopt in this study enables efficient analysis of a wide range of engineering problems, encompassing domains such as robotics, automotive, power systems, and more. This formulation can be also applied to the \gls{ACC} system, which is the focus of our investigation in this paper. In the subsequent section, we will proceed to model the impact of a \gls{DoS} attack on this particular system.

\subsection{Attack Model}
Regarding the fact that the attacks are implemented in the cyber layer, one needs to take into account the interactions in the discrete space. Therefore, let us consider the Euler approximation of (\ref{system}) under actuator attack and with measurement noise as the following
\begin{align}
    \begin{cases}x^{t+1}=A_\tau x^{t}+\Delta(x^{k})\tau+B\tau u^{t}+d^t, & t \notin \alpha\\
    x^{t+1}=A_\tau x^{k}+\Delta(x^{k})\tau+d^t, & t\in 
    \alpha\label{Attack_lin}\\
    \end{cases}\\
    y^t=C(I+\chi^t)x^t+\epsilon^t, \qquad t\in \{0,1,\dots\}\qquad\label{Attack_lin_y}
\end{align}
where $\tau$ is the sampling time and $A_\tau=A\tau+I$ and $B_\tau=B\tau$ are discretized system matrices. Moreover, we assume that we can measure all the states, i.e. $C=I$. Then, the term $d^t$ lumps together the discretization error and some bounded input disturbance. Moreover, we denote the set of time steps during which the \gls{DoS} attack is active using $\alpha\in S_{\alpha}$, where $S_{\alpha}$ represents the set of all possible sequences of attacks. 

Moreover, to ensure a rather practical model of the problem we take also into account the effect of noise. Therefore, the measurements are given by $y^t\in {\rm I\!R}^n$ for steps $t\in \{0,1,\dots\}$ that are affected by noise. The vector $\epsilon^t \in {\rm I\!R}^n$ and the diagonal matrix $\chi=\text{diag}([\chi_1,\dots,\chi_n]) \in {\rm I\!R}^{n\times n}$ are the bounded additive and multiplicative noises affecting the measurements, i.e $|\epsilon^t|\leq \bar \epsilon \in {\rm I\!R}^{n}_+$ and $|\chi^t|\leq \bar \chi$, where $\bar \chi$ is a diagonal matrix with positive values. The distributions of the noises applied are not necessarily uniform. In fact, the formulation can accommodate other distributions, such as truncated Gaussian noise.
\subsection{Objective}
Given the definition of the system under \gls{DoS} attack, we can define the constrained control problem which is solved at each time step $t$ in the rolling horizon fashion for the horizon length of $N$ as
\begin{align} \label{objective}
    {J^*}^t=\underset{u^t(.)}{\min} &\sum_{k=t}^{t+N-1} ({y^k}^TQ{y^k}+{u^k}^TRu^k)+{y^N}^TQ^N{y^N}, \nonumber\\
    &\text{subject to \quad (\ref{Attack_lin}),\text{ }(\ref{Attack_lin_y})}, \nonumber\\
    &\hspace{18mm} x^k \in D, \text{ and } u^k\in U,
\end{align}
for all sequences of attacks $\alpha \in S_\alpha$, noises, and disturbances within their sets of definitions. The objective defined includes stage and terminal costs, respectively.

Addressing the uncertainties inherent in the model, it becomes apparent that a direct approach to solving the optimization problem is not viable. In light of this, the subsequent section explores an alternative technique that can effectively handle the problem by transforming it into the standard form, incorporating parametric and additive uncertainties. This approach capitalizes on the existence of efficient techniques specifically designed to tackle such formulations.

\section{\uppercase{Resilient Framework}}
\label{sec:Resilient Framework}
 In this section, we present the components required for establishing the proposed resilient control in detail. Having the model of the attack defined, we first derive an equivalent uncertain model that facilitates efficient analyses. Second, we present the tightening-based solutions for addressing this problem. Finally, we summarize the entire framework by presenting two algorithms that encapsulate the proposed resilient control approach.

\subsection{Equivalent Uncertain Model}
The following lemma provides regulation for the lumped disturbance present in the model. 
\begin{lemma} \label{lemma_bound}
    The disturbance term $d^t$ is bounded by $\bar d \in {\rm I\!R}^+$.
\end{lemma}
\begin{proof}
    Considering the Assumption \ref{assu_lip} and compact domains, it can be shown that the local truncation error resulting from the discretization remains bounded for all $(x^t,u^t)\in D\times U$. Moreover, according to our assumption, the system may be prone to some bounded input disturbance. Therefore, the lumped disturbance $d^t$ is also bounded by some $\bar d \in {\rm I\!R}^{n}_+$. 
\end{proof}

Assuming a periodic DoS attack, as a well-known class of attacks \cite{cetinkaya2019overview}, we can rewrite the system by averaging both modes of (\ref{Attack_lin}) as
\begin{align} 
    &x^{t+1}=A_\tau x^{t}+\Delta(x^{t})\tau+B_\tau  u^t(1-\omega^t)+d^t,\label{Attack_lin_}
\end{align}
where $\omega^t \in [0,1]$ takes continuous values in this closed interval representing the intensity of the DoS attack.

\begin{assumption} \label{assu_attack_bound}
    The attack signals $\omega^t$ is bounded and the estimated values of the upper bounds are known at each time step, i.e. $\exists \bar \omega\in{\rm I\!R}^+$ such that $|\omega^t|\leq \bar \omega$ for $k \in \{0,\dots, N-1\}$.
\end{assumption}
Assumption \ref{assu_attack_bound} automatically holds for the type of attack considered and the problem formulation where a worst case of $\omega^t$ values is given by $1$. However, in practice, based on the estimations of the attack intensity, smaller values than $1$ may be considered for $\bar \omega$ at each step $t$.

In what follows, we investigate how the measurements deviate from the predictions given by the nominal dynamics
\begin{align*}
  \bar F(x^t,u^t)=A_\tau x^{t}+B_\tau  u^{t}. 
\end{align*}
Therefore, consider
\begin{align} \label{proof_eq1}
    y^{t+1}&-\bar F(x^t,u^t)\nonumber\\
    &=(I+\chi^{t+1})x^{t+1}+\epsilon^{t+1}-A_\tau x^{t}-B_\tau  u^{t}\nonumber\\
    &=(I+\chi^{t+1}) \Big(A_\tau x^{t}+\Delta(x^{t})\tau+B_\tau  u^{t}(1-\omega^t)+d^t\Big)\nonumber\\
    &\quad+\epsilon^{t+1}-A_\tau x^{t}-B_\tau  u^{t}\nonumber\\
    &=\chi^{t+1}A_\tau x^{t}+\big((I+\chi^{t+1})(1-\omega^t)-I\big)B_\tau  u^t\nonumber\\
    &\quad+(I+\chi^{t+1}) \big(\Delta(x^{t})\tau+d^t\big)+\epsilon^{t+1}\nonumber\\
    &=\chi^{t+1}A_\tau x^{t}+(\chi^{t+1}-(I+\chi^{t+1})\omega^t)B_\tau  u^t\nonumber\\
    &\quad +(I+\chi^{t+1})\Delta(x^{t})\tau+(I+\chi^{t+1}) d^t+\epsilon^{t+1},
\end{align}
where we used (\ref{Attack_lin_}) in the derivations.
 Starting with the first two terms, let us define the convex polytopic sets $\Pi_A$ and $\Pi_B$ as below that contain the uncertainty corresponding to $A_\tau$ and $B_\tau$ matrices of the nominal dynamic, respectively, 
\begin{align} \label{polytop}
    &\Pi_A = \text{conv}(\{\chi^{t+1}A_\tau |\chi^{t+1} \in \chi_{v}\}),\nonumber \\
    &\Pi_B = \text{conv}(\{(\chi^{t+1}-(I+\chi^{t+1})\omega^t)B_\tau | \chi^{t+1} \in \chi_{v},\nonumber\\
    &\qquad\qquad\omega^t\in \{0,\bar \omega\}\}),
\end{align}
for all vertices $\chi_{v}$ given by the extreme values of $\chi^{t+1}$. Regarding that $\bar \chi$ is diagonal, $\chi_{v}$ can be easily calculated.

In the subsequent step, the remaining terms are treated as additive uncertainty. It is important to highlight that, in order to obtain specific bounds for each system dynamic specifically, component-wise calculations are employed, rather than considering a norm-based approach. In this regard, the non-scalar bounds defined and the Lipschitz constants in Assumption \ref{assu_lip} facilitate these computations. Hence, we aim for a bound using equation (\ref{proof_eq1}) that yields
\begin{align}
    |y^{t+1}&-(\bar F(x^t,u^t)+\chi^{t+1}A_\tau x^t\nonumber\\
    &+(\chi^{t+1}-(I+\chi^{t+1})\omega^t)B_\tau  u^t)|\nonumber\\
    &=|(I+\chi^{t+1})\Delta(x^t)\tau+(I+\chi^{t+1}) d^t+\epsilon^{t+1}| \nonumber \\
    &\leq |(I+\chi^{t+1})\Delta(x^t)\tau|+|(I+\chi^{t+1}) d^t|+|\epsilon^{t+1}| \nonumber\\
    &\leq |(I+\chi^{t+1})\tau||\Delta(x^t)|+|(I+\chi^{t+1}) d^t|+|\epsilon^{t+1}|\nonumber\\
    &\leq |(I+\chi^{t+1})\tau|\eta \lVert x^t\rVert+|(I+\chi^{t+1}) d^t|+|\epsilon^{t+1}|\nonumber\\
    &\leq |(I+\bar \chi)\tau|\lVert x^t\rVert\eta +(I+\bar\chi) \bar d+\bar \epsilon,
\end{align}

where we used Assumption \ref{assu_lip} to derive the last two steps. This provides a bound for the remaining terms while one can use $\underset{x^t \in D}{\text{max}}(\lVert x^t \rVert)$ to bound $\lVert x^t \rVert$. However, it may not offer a useful bound if $\eta$ is not small. As an alternative, we can set different values for the bounds based on the current value of $x^t$, instead. In this case, considering that we do not measure $x^t$ exactly, we can employ the measurements through (\ref{Attack_lin_y}) to obtain
\begin{align}
    \lVert x^t\rVert&=\lVert(I+\chi^{t+1})^{-1}\rVert\lVert(y^t-\epsilon^{t+1})\rVert \nonumber\\
    &\leq \lVert(I+\chi^{t+1})^{-1}\rVert(\lVert y^t\rVert+\lVert\epsilon^{t+1}\rVert)\nonumber\\
    & \leq\lVert(I-\bar\chi)^{-1}\rVert(\lVert y^t\rVert+\lVert\bar\epsilon\rVert).
\end{align}

We summarize the computations by utilizing a discrete-time linear model that incorporates parametric and additive uncertainty. This model serves as an overapproximation of the continuous-time dynamics (\ref{system}) in the presence of a DoS attack and uncertainty,
\begin{align} \label{uncertain_linear}
    x^{t+1}=(A_\tau +\hat \Delta_A)x^t+ (B_\tau+\hat \Delta_B) u^t + \hat d^t,
\end{align}
where $|\hat d^t| \leq \hat \delta$ with
\begin{align} \label{d_bound}
\hat \delta&=|(I+\bar \chi)\tau|\lVert(I-\bar\chi)^{-1}\rVert(\lVert y^t\rVert+\lVert\bar\epsilon\rVert)\eta \nonumber \\
&\quad+(I+\bar\chi) \bar d+\bar \epsilon,
\end{align}
 $\hat \Delta_A \in \Pi_A $ and $ \hat \Delta_B \in \Pi_B$, with defined $\Pi_A$ and $\Pi_B$ by (\ref{polytop}).

\subsection{Resilience via Constraint Tightening}
In this section, we summarize the constraint tightening technique employed for solving the \gls{RMPC} problem. Accordingly, we deliver the resilient framework proposed using also the results obtained in the previous section.

 % In \cite{bujarbaruah2021simple}, a hybrid approach to constraint tightening in robust \gls{MPC} is introduced. They divide the process into two cases based on the prediction horizon length by exactly solving the \gls{MPC} problem for a horizon length of one. For longer horizons, they incorporate model uncertainty as a net-additive component and adjust constraints based on worst-case bounds. By employing an adaptive horizon strategy and constructing appropriate terminal sets and costs, they achieve recursive feasibility and input-to-state stability.

 Regarding the model in (\ref{uncertain_linear}), although utilizing fixed bounds can be effective for addressing slight model uncertainty and noise effects, it may not adequately capture the impact of \gls{DoS} attacks, which is considered the major source of uncertainty in the model. To address this, we propose a more adaptable approach that can accommodate various strengths of attacks while maintaining satisfactory system performance. Moreover, as previously suggested, the Lipschitz values may be large leading to large additive bounds in (\ref{d_bound}) that can be also addressed by a similar approach.
 
In order to overcome these limitations, let us define different quantities of bounds for $\omega^t$ and $ \lVert y^t \rVert$ as $[\omega]_q$ and $ [\lVert y \rVert]_q$ that are taken from a set of grid points, i.e. $([\omega]_j,[\lVert y \rVert]_i) \in G[0,\bar \omega]\times G[0,\text{sup}(\lVert y \rVert)]$ for $j=1,\dots,N_\omega$ and $i=1,\dots,N_d$, where $N_\omega$ and $N_d$ are the numbers of grid points for each dimension. Then, by online observations, one needs to ensure that the conditions $\omega^t\leq [\omega]_q $ and $\lVert y^t \rVert \leq [\lVert y \rVert]_q$ hold by choosing a suitable $q$ from the set of indices $\{1,\dots, N_\omega\times N_d\}$.
 
 By employing a similar scheme as \cite{bujarbaruah2021simple}, we consider an adaptive prediction horizon where at each time step, we solve the problem for different horizon lengths $N_t\in \{1,\dots,N\}$, and proceed with the one with the least cost. However, there is a key distinction in our approach as we take into account a collection of uncertain models, which are defined by different bounds corresponding to varying levels of attack intensity. Accordingly, we apply one of the following two approaches in handling the uncertainties depending on the value of $N_t$. 
 
 \begin{itemize}

 \item Case {$N_t=1$}: Accordingly, the robust MPC problem is exactly solved for a horizon length of one. For this purpose, assuming that there exists a feedback gain $K_q$ such that $(A_\tau +\hat \Delta_A)+ (B_\tau+\hat \Delta_{Bq})K_q$ is stable for all  $\hat \Delta_A \in \Pi_A $ and $ \hat \Delta_{Bq} \in \Pi_{Bq}$, we can construct the terminal sets $X^N_{q}$ as the maximal robust positive invariant set for $x^{t+1}=\Big((A_\tau +\hat \Delta_A)+ (B_\tau+\hat \Delta_{Bq})K_q\Big)x^t+\hat d^t$, with $q\in\{1,\dots, N_\omega\times N_d\}$. Therefore, in addition to the state and control input constraints, we need to satisfy the condition 
 %$X_N=\{x|H_N^x x\leq h^x_N\}$
 \begin{align}
     \Big((A_\tau +\hat \Delta_A)+ (B_\tau+\hat \Delta_{Bq})K_q\Big)x^t+\hat d^t \in X^N_{q} 
 \end{align}
 in the optimization problem, where $X^N_q$ is a convex set defining the terminal set for the model given by the index $q$. It should be noted that, $\Delta_{Bq}$ and $|\hat d^t|\leq \hat \delta_q$ are characterized by the quantities $[\omega]_q$ and $ [\lVert y \rVert]_q$ for given index $q$, according to relations (\ref{polytop}) and (\ref{d_bound}).

\item Case {$N_t>1$}: Given the computational intensity of the method employed in the previous case for multi-step predictions, an alternative approach is taken. Bounds are utilized to over-approximate system uncertainty, rather than precise calculations by using the technique found in \cite{goulart2006optimization}. This allows for the treatment of all uncertainties as a net-additive component, utilizing a more constructive technique. The adoption of this approach aims to mitigate the computational burden while still effectively accounting for system uncertainties. 

 \end{itemize}
The presented resilient framework can be effectively implemented in two parts. In the first part, the model and uncertainty bounds are processed to characterize the constraints. These computations are conducted offline in advance which facilitates the preparation of constraints. By performing these computations beforehand, the constraints can be readily available for subsequent utilization. In Algorithm \ref{alg_offline}, which presents the offline procedure, we grid the space $G[0,\bar \omega]\times G[0,\text{sup}(\lVert y \rVert)]$ to obtain different values $[\omega]_q$ and $ [\lVert y \rVert]_q$. Then, we use (\ref{polytop}) and (\ref{d_bound}) to calculate the corresponding bounds for all $q$.

\begin{algorithm}[h] 
 \caption{Offline computations of the bounds and state constraints.}
 \vspace{0mm}
 \label{alg_offline}
 \KwData{System matrices, $\bar \chi$, $\bar \epsilon$, domain and control constraints $\gets D$ and $U $}
\vspace{-2mm}
 \KwResult{\begin{tabular}{@{\hspace*{0.1em}}l@{}}
        \\$\hat \Delta_A$, $\hat \Delta_{Bq}$, {\normalfont and} $\hat \delta_q$. \\
        Terminal sets $X_q^N$  $\forall  ([\omega]_q$ ,$ [\lVert y \rVert]_q)$\\  and weights $Q^N$.\\
      \end{tabular} }
 $\%$ \textit{Number of grid points}\\
 $N_\omega$ and $N_d$ $ \gets$ positive integers \;
 $\%$ \textit{Grid points}\\
 $\omega_{List}\gets \text{linspace}(0,\bar \omega,N_\omega)$\;
 $Y_{List}\gets \text{linspace}(0,\text{sup}(\lVert y \rVert),N_d)$\;
  \textit{$\%$Using relation (\ref{polytop})}:\\
 Calculate $\hat \Delta_A$ \;
 \For{ $([\omega]_q$ ,$ [\lVert y \rVert]_q)$ in $\omega_{List}\times Y_{List}$}{
 \textit{$\%$Using relation (\ref{polytop})}:\\
  Calculate  $\hat \Delta_{Bq}$\;
  \textit{$\%$ Using relation (\ref{d_bound})}:\\
   Calculate $\hat \delta_q$ \;
   \textit{$\%$ Employing \cite{bujarbaruah2021simple}}:\\
   Calculate the terminal sets $X_q^N$\;
   Calculate $Q^N$ \;
   }
\end{algorithm}

The second part of the implementation involves the utilization of the pre-determined constraints in an online optimization-based control approach. During the online phase, these prepared constraints are incorporated into an optimization framework to generate real-time control actions by exploiting \cite{bujarbaruah2021simple}. For this purpose, we obtain an estimation of the intensity of an ongoing attack using an anomaly detection method and choose applicable $[\omega]_q$ and $ [\lVert y \rVert]_q$. By integrating their corresponding constraints into the optimization process, the control approach ensures that the system operates within desired limits while effectively addressing uncertainties. This procedure is summarized in Algorithm \ref{alg_online}.

 \begin{remark}
    The offline computations enable the efficient characterization of constraints, resulting in computational time savings during online implementation. Furthermore, using a hybrid technique employing the two cases according to $N_t$ facilitates a responsive optimization process, thereby potentially enabling realtime resilient control actions in real-world applications. 
 \end{remark}
 
\begin{algorithm}[h] 
 \vspace{-2mm}
 \caption{Online computation of resilient control value.}
  \label{alg_online}
 \KwData{$A_\tau$, $B_\tau$, $\hat \Delta_A$, $\hat \Delta_{Bq}$, {\normalfont and} $\hat \delta_q$ $\forall ([\omega]_q$ ,$ [\lVert y \rVert]_q)$.\\
        $X_q^N$  $\forall ([\omega]_q$ ,$ [\lVert y \rVert]_q)$ and weights $Q^N$.
      }
 \KwResult{$u^t$ }
 \%initialization\\
 $N\gets$ positive integer\;
 \For{$t= 0, 1, \dots$}{
   Measure $y^t$\;
   Detect attack\;
   Estimate $\hat\omega^t$\ based on attack detected\;
   Choose $([\omega]_q$ ,$ [\lVert y \rVert]_q)$ based on $\hat \omega^t$ and $y^t$\;
   \For{$N_t= 1, 2, \dots,N$}{
   Solve RMPC for $\hat \Delta_A,\hat \Delta_{Bq},\hat \delta_q$, $X_q^N$,$Q^N$ \;
   }
   Apply $u^t$: arg $\underset{N_t}{min} $ ${J^*}$\; }
\end{algorithm}
 \vspace{-5mm}
 \section{\uppercase{Anomaly Detection}}
\label{sec:Anomaly Detection}
% todo3
To implement \gls{RMPC}, we need to observe the intensity of the launched attack. Therefore, the Kalman filter \cite{7947172,6859155} technique is utilized to detect the launched attack. Accordingly, the system states are estimated, and the resulting residuals are employed to detect the attack. 
 
The residual at $t$th step is defined as
\begin{equation}
    r = y^{t}-y^{t|t-1}
\end{equation}
s.t.
\begin{equation}
    r = y^t - C\hat{x}^{t|t-1}.
\end{equation}

Then, according to \cite{5718158,7040293}, we can determine if the system is under attack. 
\begin{align}\label{eq:detect attack}
    \begin{cases}t\notin 
    \alpha, &\text{if}\ r^T\mathcal{P}^{-1} r\leq T,\\
    t\in 
    \alpha, & \text{if}\ r^T\mathcal{P} ^{-1}r> T,\\
    \end{cases}
\end{align}
where $T\in {\rm I\!R}_{+}$ is the threshold of the corresponding residue, to be tuned by the user. Moreover, $\mathcal{P}$ represents the covariance matrix of the residue $r$, and $T$ is the threshold. According to (\ref{eq:detect attack}), we can determine whether the system is under attack.

In the next section, the detection results are presented in more detail. For improved results, robust Kalman filters \cite{10056247} can be exploited alternatively. 

\section{\uppercase{Simulation Results}}
\label{sec:Simulation Results}
To demonstrate the efficacy of the proposed approach under \gls{DoS} attacks, we validate Algorithm \ref{alg_offline} and \ref{alg_online} on the \gls{ACC} problem. Moreover, we present the details of the detection procedure and discuss some comparison results. All the simulations are conducted in Matlab on the Windows operating system with the hardware configuration of AMD Ryzen 9, 16-Core, 3.40 GHz, and 64GB of RAM.

As shown in Fig. \ref{fig:acc_example}, \gls{ACC} system consists of two vehicles, one of which is the ego vehicle and the other is the lead vehicle. The control objective defined for the ego vehicle is to maintain the distance from the lead vehicle while satisfying the state and control constraints. 

To describe the model, we employ the state variables $x=[\delta d, \delta v, \dot v_h]^T$ defined as the followings:
\begin{itemize}
    \item $\delta d$ is defined as the distance error which is the difference between the actual distance $d$ and the desired distance $d_r$ from the lead vehicle, i.e $\delta d= d-d_r$.
    \item $\delta v$  denotes the velocity difference between the lead vehicle $v_p$ and the ego vehicle $v_h$.
    \item $\dot{v}_h$ represents the acceleration of the ego vehicle.
\end{itemize}
According to \cite{takahama2018model,al2021model}, the longitudinal dynamics of the ego vehicle is given as

\begin{equation}
    \begin{aligned}\label{dynamic_cnt}
        &\dot{v}_{h} =A_fv_h+B_fu , \\
        &a_{f} =C_f v_h ,
    \end{aligned}
\end{equation}
where $a_f$ is the traction force of the vehicle converted to acceleration
\begin{equation}
    \begin{aligned}
        &A_f=-\frac{1}{T_{eng}},&B_f=-\frac{K_{eng}}{T_{eng}},\text{ and}\hspace{2mm} C_f=1.
    \end{aligned}
\end{equation}

Furthermore, $T_{eng}$ is the constant of acceleration of the engine, and $K_{eng}$ is the gain of the engine.

In addition, the reference distance is defined with respect to the velocity of the ego vehicle using
\begin{align} \label{ref_distance}
    d_r=T_{hw}v_h+d_0
\end{align}
where $T_{hw}$ is the constant time headway, and $d_0$ is the safety clearance when the lead vehicle comes to a full stop. However, for simplicity, it is assumed that the lead vehicle has some positive velocity, hence, $d_0=0$ can be used safely. 

According to (\ref{dynamic_cnt}), (\ref{ref_distance}), and the defined state vector $x$, we can obtain the discrete-time model as  (\ref{Attack_lin}) with $A_\tau =I+\begin{bmatrix}0&1&-T_{hw}\\ 0&0&-1\\ 0&0&A_f\end{bmatrix}\tau$, $B_\tau =\begin{bmatrix}0\\ 0\\ B_f\end{bmatrix}\tau$ and $C=\begin{bmatrix}1&0&0\\ 0&1&0\\ 0&0&1\end{bmatrix}$.

For more details, we refer the readers to \cite{takahama2018model,al2021model}. Also, similar to these works, we set the values of the parameters $T_{hw}$, $T_{eng}$, and $K_{eng}$, as shown in Table \ref{tab:params for acc}.

\begin{table}[h]
\caption{Parameters of \gls{ACC} model.}\label{tab:params for acc} \centering
\begin{tabular}{|c|c|c|}
  \hline
  $T_{hw}$ & $T_{eng}$ & $K_{eng}$ \\
  \hline
  1.6 & 0.46 sec  & 0.732\\
  \hline
\end{tabular}
\end{table}
\subsection{Detection}
In the simulated attack scenario, we generate a periodic attack characterized by varying active lengths for each period. Fig. \ref{fig:detection_result} illustrates the attack signal profile that initiates at  $t=2$ seconds. 

To detect this attack, we employ the Kalman filter approach, with noise covariances $Q_f=diag(0.01,0.01,0.01)$, $R=diag(0.01,0.01,0.01)$, the sampling time $t_{sample}=0.01$, and the initial covariance matrix $P_{init}=diag(0.1,0.1,0.1)$. The residuals obtained from the detection process are also depicted in the same figure, showcasing the impact of the attack events. By appropriately setting threshold values, we are able to successfully detect the attack, as demonstrated in Fig. \ref{fig:detection_result}. It is important to note that while some false positive and negative cases may occur, the Kalman filter detector is effective in detecting \gls{DoS} attacks in the majority of instances.

The bottom graph in Fig. \ref{fig:detection_result} shows the results for the estimation of attack intensity, i.e. $\hat\omega^t$, together with the exact signal $\omega^t$, which is obtained based on the true attack signal. In fact, we use a backward-moving average of the attack signal to generate such an estimation and it illustrates how effectively one can reconstruct a signal representing the attack intensity. By comparing the estimated attack strength $\hat \omega^t$ with the exact signal, we observe a close correspondence between the two. This demonstrates the ability of the estimation process to capture and reconstruct the features of the attack signal that can be used to establish the resilient controller accordingly.

  \begin{figure}[ht]
    \centering
    %trim=left botm right top
    \includegraphics[clip, trim=0.5cm 0.1cm 0.5cm 0.5cm, width=0.95\linewidth]{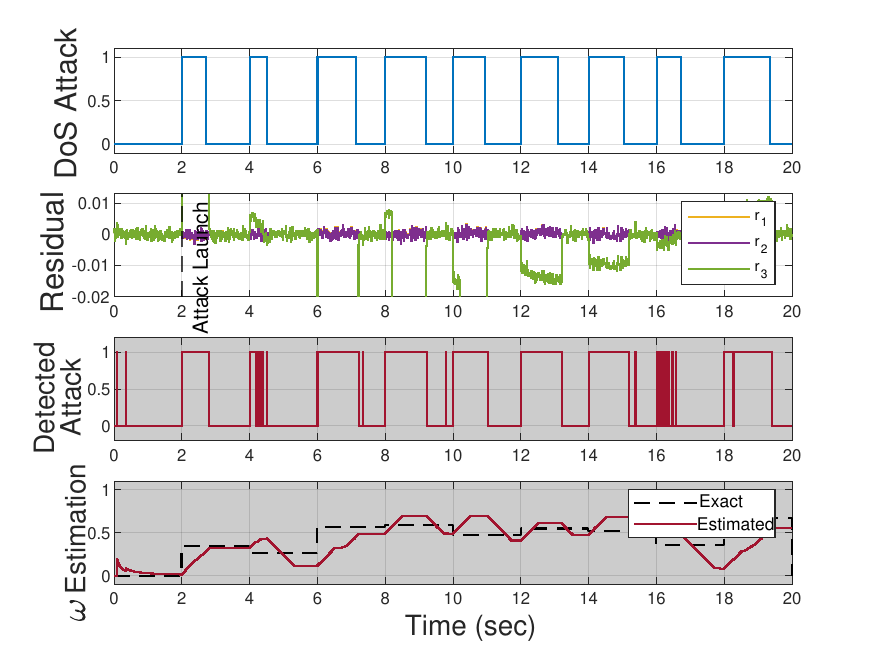}
    \caption{The result of anomaly detection together with the estimation of the strength of attack $\omega^t$.}
    \label{fig:detection_result}
\end{figure}

\subsection{Control}
% todo, how we implement RMPC?
For validation, we apply the resilient control approach proposed on the \gls{ACC} system specified. For the control computations, we consider a discretized system $\tau=0.2$ sec. To characterize the objective function (\ref{objective}), we set $Q=diag(10,10,10)$, $R=2$, and $N=5$ for all simulation if not mentioned otherwise. Moreover, the state and control are constrained within intervals $[-100,-100,-100]^T<x^t<[10,100,100]^T$ and $-20<u^t<20$, respectively. 

For the preparation of constraints, we employ the proposed Algorithm \ref{alg_offline} in Matlab, with $N_\omega=10$, $N_d=3$,$|\chi|\leq diag([1 ,1,1])\times10^{-2}$,$|\epsilon|\leq[0.1,0.1,0.1]^T$, and $\eta=[1,1,1]\times 10^{-2}$.

Having the uncertainty bounds and the terminal sets calculated in our deposit, Algorithm $\ref{alg_online}$ can be executed within the control loop to handle the attack profile discussed in the previous section. In Fig. \ref{fig:RMPC_simulation}, we present the evolutions of the \gls{ACC} system under attack, controlled by the proposed resilient controller. The control plot demonstrates that the control signal becomes zero when a \gls{DoS} attack is initiated, as it cannot be transmitted to the vehicle. Therefore, the control strategy must efficiently compensate for this absence of control during attack periods.

 \begin{figure}[h]
    \centering
    %trim=left botm right top
    \includegraphics[clip, trim=0.5cm 0.5cm 0.5cm 0.5cm, width=1\linewidth]{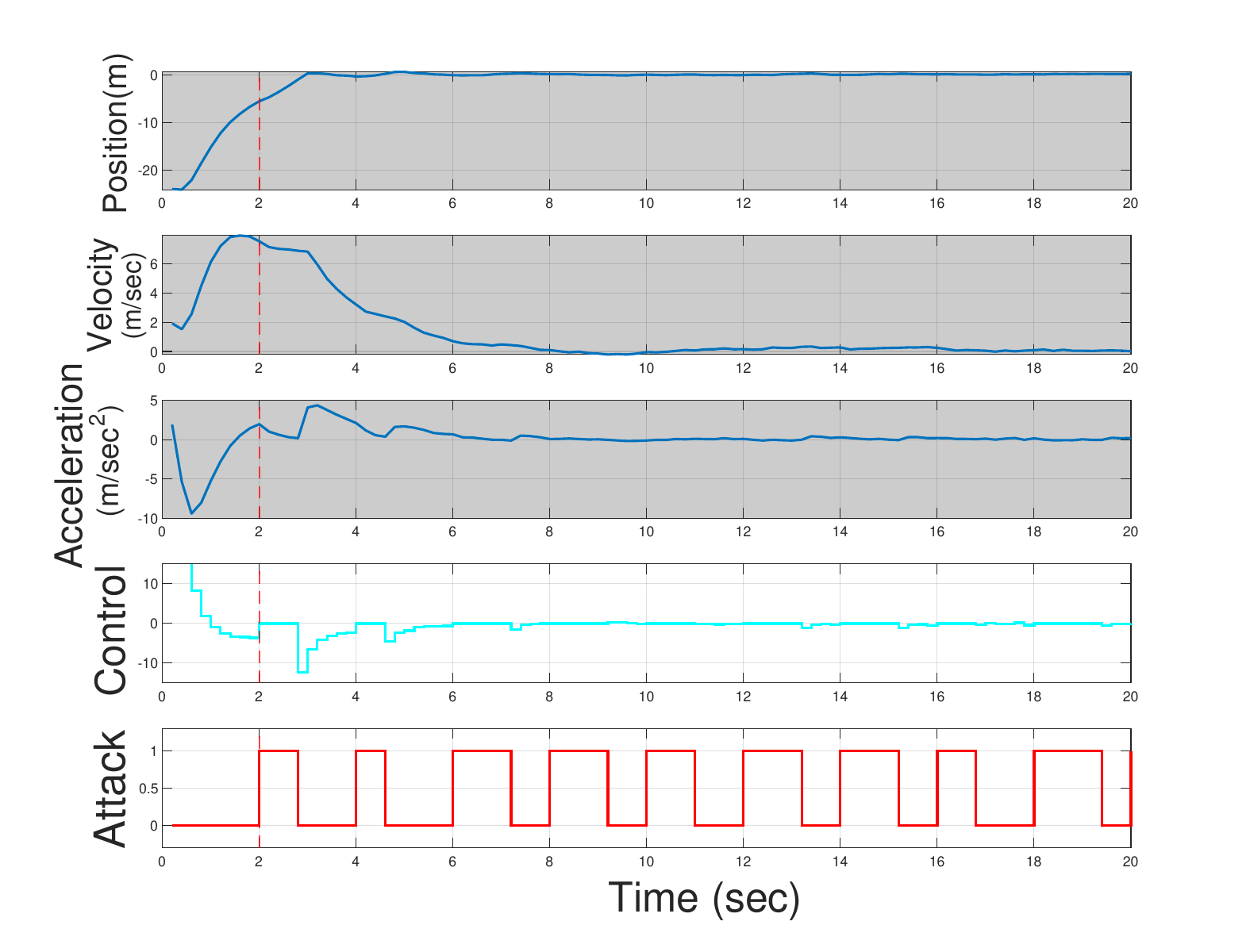}
    \caption{Details of the resilient control responses against DoS attack.}
    \label{fig:RMPC_simulation}
\end{figure}

Regarding the system states, including relative position, velocity, and acceleration, the primary objective of the control is to regulate the system to the origin, represented by $x=[0,0,0]^T$. Remarkably, despite the occurrence of the \gls{DoS} attack on the system, all the states effectively converge to zero starting from a random initial condition with the utilization of the proposed resilient control strategy. These results clearly indicate the effectiveness and resilience of the proposed control approach in managing the impact of \gls{DoS} attacks on the \gls{ACC} system. 

\subsection{Comparison Results}
For a better demonstration of the effectiveness of the proposed resilient control, we also apply the standard \gls{MPC} in the same attack scenario and \gls{ACC} system configuration. Therefore, we employ CasADi optimization library \cite{Andersson2018} in Matlab.

 \begin{figure}[h!]
    \centering
    %trim=left botm right top
    \includegraphics[clip, trim=0.5cm 0.5cm 0.5cm 0.5cm, width=1\linewidth]{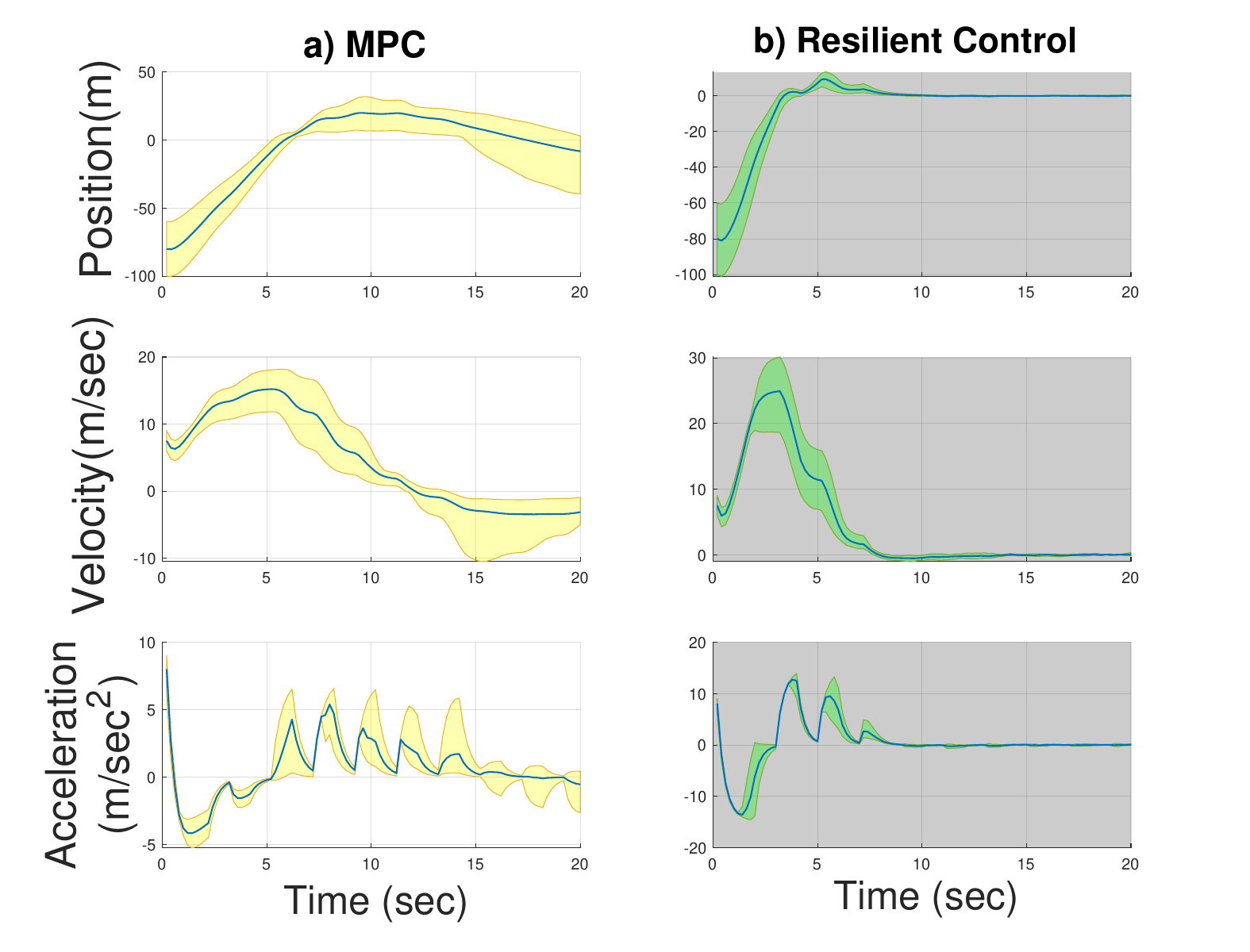}
    \caption{Comparison between \gls{MPC} and the proposed resilient control for a set of initial conditions.}
    \label{fig:state trajectory of MPC and RMPC}
\end{figure}

In Fig. \ref{fig:state trajectory of MPC and RMPC}, we illustrate the comparison between the proposed resilient control and \gls{MPC} in terms of the convergence of system trajectories. For this result, we consider a set of initial conditions, then we show the mean and upper/lower bound of the trajectories using the line curves and the shaded areas, respectively. It is evident that the proposed method results in faster convergence while \gls{MPC} fails to effectively regulate the system trajectories to zero. The performance can be also quantitatively compared by keeping track of the cost function over time, i.e. the summation in (\ref{objective}), for both techniques. These values of cost can be shown as two signals as in Fig. \ref{fig:cost comparation}. In Table \ref{tab:example1}, the numerical values
of costs are reported by averaging for all initial conditions. Accordingly, in the simulated scenario, the mean cost value shows about 38$\%$ improvement for the proposed approach in comparison to standard \gls{MPC}.

\begin{table}[h]
\caption{Comparison of mean cost values obtained.}\label{tab:example1} \centering
\begin{tabular}{|c|c|c|}
  \hline
  MPC & Resilient Control & Improvement\\
  \hline
  $9.2163\times 10^{5}$ & $5.7300 \times 10^5$  &38$\%$\\
  \hline
\end{tabular}
\end{table}
\vspace{-5mm}
\paragraph{Computational time:} 
Finally, as a comparison of the computational complexity of the proposed approach, we illustrate the runtime results for both presented and \gls{MPC} techniques in Fig. \ref{fig:runtime}. We run the proposed technique in two different configurations with $N=1$ and $N=5$. By comparing these results, the basic \gls{MPC} runs faster as expected since it solves a less complicated problem. However, the runtime results for both settings of the proposed technique are comparable to \gls{MPC}, making it a potential candidate for replacing non-resilient controllers in real-world implementations. This is mostly because the main part of the computations is done offline using Algorithm \ref{alg_offline}.

 \begin{figure}[h!]
    \centering
    %trim=left botm right top
    \includegraphics[clip, trim=0.5cm 0.1cm 0.5cm 0.2cm, width=0.8\linewidth]{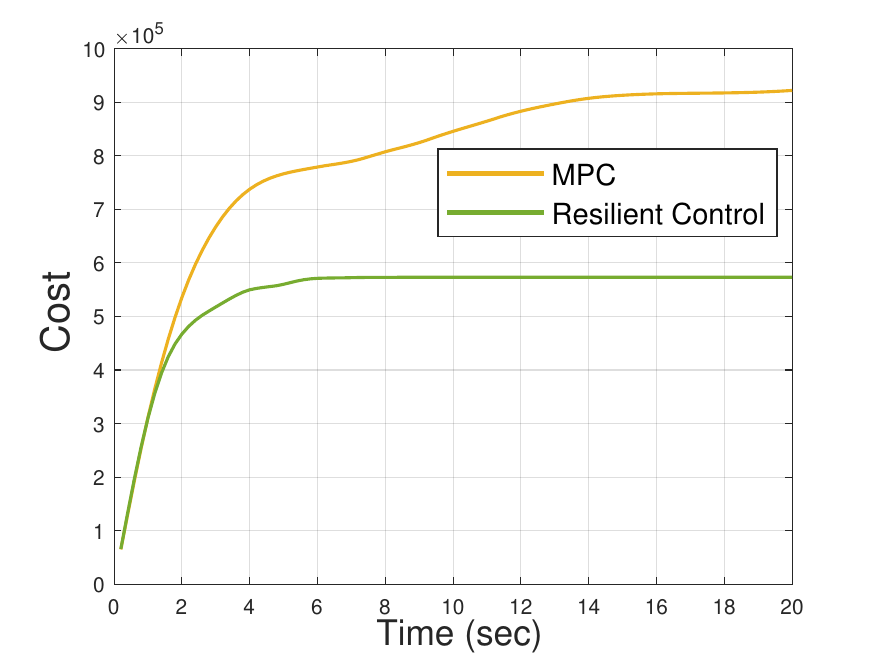}
    \caption{In this graph, by keeping track of the cost function for the proposed resilient control and \gls{MPC}, we compare the performance, where the proposed method clearly outperforms by resulting in a lower control cost.}
    \label{fig:cost comparation}
\end{figure}
 \begin{figure}[h]
    \centering
    %trim=left botm right top
    \includegraphics[clip, trim=0.5cm 0.2cm 0.5cm 0.5cm, width=0.8\linewidth]{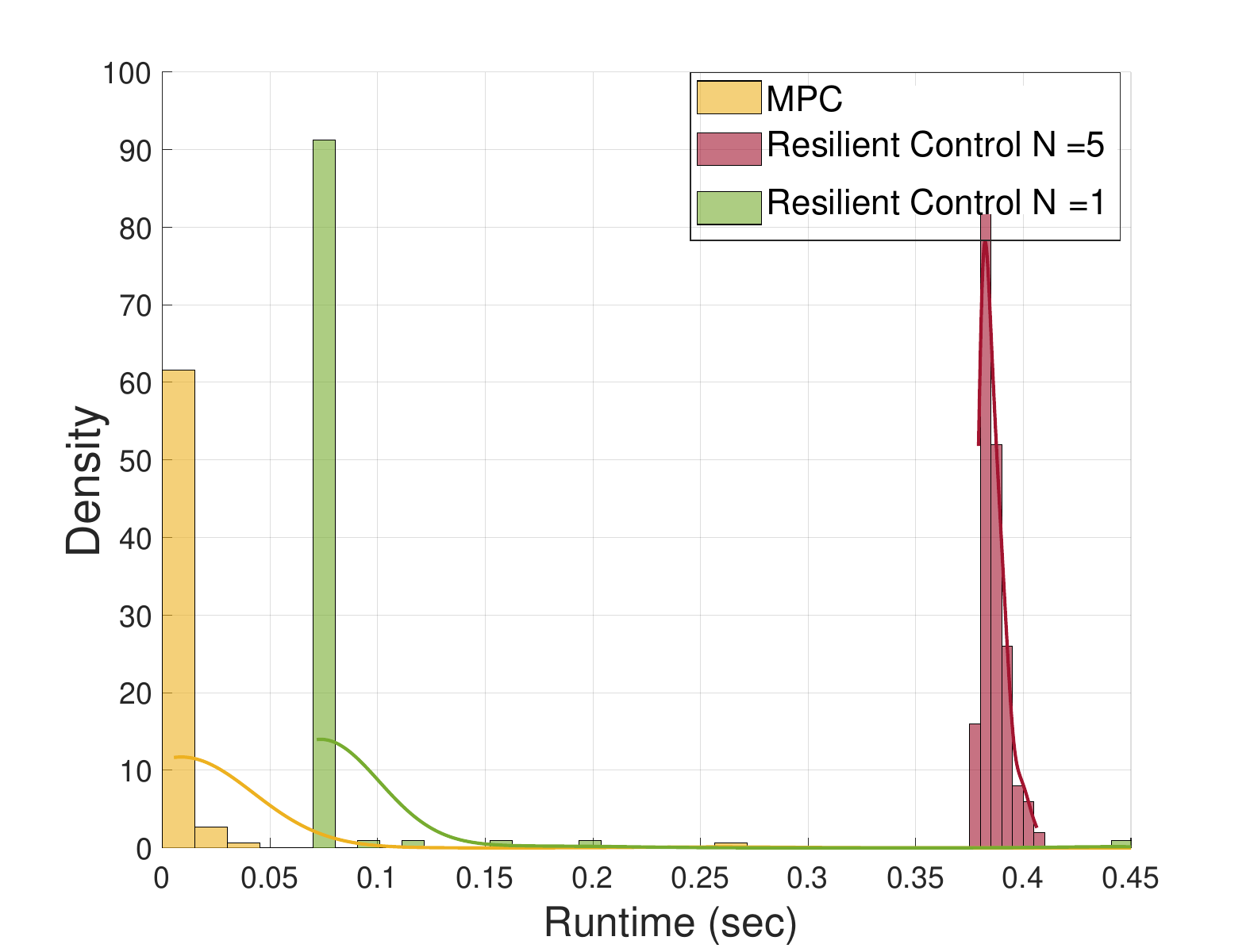}
    \caption{We compared the runtime results for both techniques: the basic MPC and two different configurations of the proposed method with $N=1$ and $N=5$.}
    \label{fig:runtime}
\end{figure}

\subsection{Conclusion}
This study proposed a novel framework for designing a resilient \gls{MPC} system to handle uncertain linear systems under periodic \gls{DoS} attacks. The DoS attack was modeled as an uncertain parameter-varying system with additive disturbance, and the Kalman filter was used for anomaly detection. An optimization-based resilient algorithm was developed using a robust constraint-tightening \gls{MPC} approach. We implemented the approach to the \gls{ACC} problem, showcasing its effectiveness in mitigating the impact of periodic attacks and ensuring system stability. Overall, the study provided a solution to enhance the resilience of control systems in the presence of \gls{DoS} attacks. Incorporating robust attack detection methods and extending the framework to encompass various types of attacks can be potentially promising for future research.

\bibliographystyle{apalike}
{\small
\bibliography{example}}

\end{document}